\documentclass[12pt,fleqn,leqno]{article}
\usepackage{amsmath,amssymb}
\usepackage{multicol}
\usepackage{amsthm}
\usepackage[]{graphicx}
\usepackage{natbib}

\allowdisplaybreaks

\numberwithin{equation}{section}

\makeatletter
\renewcommand{\section}{\@startsection{section}{1}{0pt}{20pt}{6pt}{\large\bf}}
\renewcommand{\@seccntformat}[1]{\csname the#1\endcsname.\ }

\def\footnoterule{\kern -3pt \hrule width 2.7 true cm \kern 2.6pt}

\def\vs{\vspace}

\def\EE{\mathsf E}

\def\QQ{\mathsf Q}

\def\wt{\widetilde}

\newcommand{\p}{\! +\! }
\newcommand{\m}{\! -\! }

\newtheorem{theorem}{Theorem}[section]

\newtheorem{remark}[theorem]{Remark}

\begin{document}

\title{\textbf{Closed form optimal exercise boundary of the American put option}}
\author{Yerkin Kitapbayev\thanks{Department of Mathematics, NC State University, Raleigh NC, USA; \texttt{ykitapb@ncsu.edu}}
}
\maketitle


{\par \leftskip=2.6cm \rightskip=2.6cm \footnotesize

We present three models of stock price with time-dependent interest rate, dividend yield, and volatility, respectively, that allow for explicit forms of the optimal exercise boundary of the finite maturity American put option. The optimal exercise boundary satisfies the nonlinear integral equation of Volterra type. 
We choose time-dependent parameters of the model so that the integral equation for the exercise boundary can be solved in the closed form. 
We also define the contracts of put type with time-dependent strike price that support the explicit optimal exercise boundary.
\par}




\vs{-18pt}

\vs{-18pt}

\section{Introduction}

One of the important problems in mathematical finance is the pricing of American put options. The explicit solution has been obtained in the perpetual case under   the classical Black-Scholes model. However,  in the finite maturity case the problem has not been solved in the closed form. The main obstacle is the time dependence of the optimal exercise boundary and option price. 
There were some theoretical characterizations of the option price and optimal exercise boundary in terms of, e.g., integral equation, free-boundary PDE system, series representations etc.

Several numerical approaches have been used to tackle the problem in the case of Black-Scholes model. We will mention just some of them. One of the most  popular methods in practice is the binomial tree model by \cite{CRR-1979}.
Another approach is to reduce the American option pricing problem to a free-boundary PDE system and the latter can be solved, e.g., by  the finite difference method (see \cite{BS-1977}). Then, there was a sequence of papers by 
\cite{K-1990}, \cite{J-1991}, and \cite{CJM-1992},  where the so-called early exercise premium (EEP) formula for the American option price has been derived. Based on this result, it was shown that the optimal exercise boundary satisfies the integral equation of Volterra type. \cite{Pe-2005} proved that the optimal boundary is the unique solution to this integral equation. This led to the development of numerical schemes to tackle the corresponding integral equation as it cannot be solved in the closed form. 
Another strand of the literature was devoted to the exploitation of Monte-Carlo based methods (see e.g. \cite{LS-2001} among others). Finally, several closed form approximations have been introduced (see e.g. \cite{BD-1996}, \cite{BAW-1987}).

In this paper, we employ an integral equation approach. The main argument is to deviate from the standard model with constant parameters and choose the deterministic functions for the interest rate (or dividend yield) to obtain the integral equation that can be solved in the closed form. 
Hence, we provided the extension of the standard Black-Scholes model such that the American put option has an explicit optimal exercise policy.
We also specify the extension with time-dependent  volatility function that allows us to determine the optimal exercise boundary as the solution to a simple algebraic equation. 
Finally, we consider the case of the standard model with constant parameters but the contract with time-dependent strike.
We select the latter so that the exercise policy can be found in the closed form. 

Our paper is somewhat related to the recent work by \cite{KS-2019} where so-called inverse optimal stopping problems have been solved. They considered
general diffusion framework, and the goal was to determine a time-dependent function $\pi(t)$ that is added to the original payoff in order to obtain the desired stopping rule. Some general existence  results
and representation formulas for $\pi(t)$ were established. 
The difference between our paper and \cite{KS-2019} is that we mostly seek the model parameters instead of incremental payoff function $\pi(t)$ and also by considering a particular case of American put option under extension of Black-Scholes model we obtained explicit results.

Finally, we will briefly discuss possible applications of the results and directions for future research.
One way to employ the representations in this paper is to consider them as the closed form approximations of the American put option price under the standard model. 
Second, the finite maturity stopping problems can be found in other mathematical finance applications, e.g. defaultable debt pricing, real options, mortgage pricing etc. Due to intractability of finite horizon problems, the arguments of this paper can be used to create examples with closed form solution to these problems. 
Third, one can try to extend the arguments of this paper to stochastic volatility or jump models with some time-dependent parameters so that explicit optimal exercise rules- can be achieved.

This paper is structured as follows. In Section \ref{rev} we review the well-known results for American put option under the standard model with constant parameters.  Sections \ref{int}, \ref{dividend}, \ref{volat} and \ref{strike} present expressions for interest rate, dividend yield, volatility and strike functions, respectively, that allow for the optimal exercise boundaries in the closed form.

\bigskip
\section{Standard model: Review}\label{rev}

In this section we briefly review known results for the American put option problem under geometric Brownian motion model with constant parameters. Let us assume that the dynamics of the asset price $X$ under the risk-neutral measure  $\QQ$ is given by
\begin{equation}
dX_t=(r-\delta)X_tdt+\sigma X_t dW_t
\end{equation}
for $t>0$, where $r>0$ is the interest rate, $\delta\ge 0$ is the dividend yield,   $\sigma>0$ is the constant volatility, and $W$ is a standard Brownian motion (SBM) under $\QQ$.
It is well known that to determine the price $V(t,x)$ of American put option on $X$ at time $t$ with strike $K$ and maturity $T>0$, one can solve the following optimal stopping problem
\begin{equation}\label{prob-interest}
V(t,x)=\sup_{t\le \tau\le T}\EE_{t,x}\left[e^{-r(\tau-t) }(K - X_\tau)^+\right]
\end{equation}
where $\EE_{t,x}$ is the expectation under $\QQ$ given that 
$X_t=x$, and the supremum is taken over all $\mathcal{F}^X$-stopping times $\tau$. 

Another well known result is that the optimal exercise policy $\tau^*$ can be described by the optimal exercise boundary $b(t)$ such that
\begin{equation}
    \tau^*=\inf\{s\in [t,T]:X_s\le b(s) \}.
\end{equation}
Hence, the goal is to find the pair $(V,b)$. 
\medskip

It has been shown in several papers, e.g., \cite{K-1990}, \cite{J-1991},
\cite{CJM-1992},  that the exercise boundary satisfies the nonlinear integral equation of Volterra type
\begin{align}\label{IE}
K-b(t)=&V^e(t,b(t))\\
&+rK\int_t^T e^{-r(u-t)}N\left(\frac{1}{\sigma\sqrt{u-t}}\left(\log \frac{ b(u)}{ b(t)}- \left(r\m\delta\m\frac{\sigma^2}{2}\right)(u-t)\right)\right)du\notag\\
&-\delta b(t)\int_t^T e^{-\delta(u-t)}N\left(\frac{1}{\sigma\sqrt{u-t}}\left(\log \frac{ b(u)}{ b(t)}- \left(r\m\delta\p\frac{\sigma^2}{2}\right)(u-t)\right)\right)du\notag
\end{align}
for $t\in [0,T)$ with $b(T-)=K\cdot\min(1,r/\delta)$, where 
$N$ is the cdf of $\mathcal{N} (0,1)$ and $V_e$ is the European option price
\begin{align}
V^e(t,x)= &Ke^{-r(T-t)}N\left(\frac{1}{\sigma\sqrt{T-t}}\left(\log \frac{ K}{x}\m \left(r\m\delta\m\frac{\sigma^2}{2}\right)(T-t)\right)\right)
\\
&-xe^{-\delta(T-t)} N\left(\frac{1}{\sigma\sqrt{T-t}}\left(\log \frac{ K}{x}-\left(r\m\delta\p\frac{\sigma^2}{2}\right)(T-t)\right)\right)\notag
\end{align}
for $t\in[0,T)$ and $x>0$. 
Later it was proven by \cite{Pe-2005} that the optimal exercise boundary $b$ is the unique solution to the equation \eqref{IE} in the class of continuous functions. To the best of our knowledge, this equation does not have explicit solution but it can be solved numerically by backward induction and using some quadrature scheme. 
Once the boundary $b$ is obtained, the American option price can be computed using so-called early exercise premium (EEP) representation formula
\begin{align}\label{EEP}
V(t,x)=&V^e(t,x)
+\EE_{t,x}\left[\int_t^T e^{-r(u-t)}(rK-\delta X_u)I(X_u\le b(u))du\right]\notag\\
=&V^e(t,x)\notag\\
&+rK\int_t^T e^{-r(u-t)}N\left(\frac{1}{\sigma\sqrt{u-t}}\left(\log \frac{ b(u)}{ x}- \left(r\m\delta\m\frac{\sigma^2}{2}\right)(u-t)\right)\right)du\notag\\
&-\delta x\int_t^T e^{-\delta(u-t)}N\left(\frac{1}{\sigma\sqrt{u-t}}\left(\log \frac{ b(u)}{x}- \left(r\m\delta\p\frac{\sigma^2}{2}\right)(u-t)\right)\right)du\notag
\end{align}
for $t\in[0,T)$ and $x>0$. It essentially decomposes the American option price as the sum of European option price and EEP. 
The latter takes into the account the local benefits $(rK-\delta X_u)$  of early exercise  at time $u$ given that it is optimal to stop, i.e., $I(X_u\le b(u))$. Then these benefits are discounted to time $t$ and integrated over the interval $[t,T]$. The second equality follows from Fubini's theorem and the distribution of $X_u$. 


\section{Time-dependent interest rate}\label{int}

1. Let us consider American put option problem with strike $K$ and maturity $T$. We assume that the dynamics of the asset price $X$ under the risk-neutral measure $\QQ$  is given by
\begin{equation}
dX_t=(r(t)\m \delta)X_tdt+\sigma X_t dW_t
\end{equation}
with  deterministic interest rate $r(t)>0$ (to be specified), constant dividend yield $\delta\ge 0$ and constant volatility
$\sigma>0$.  The goal of this section is to provide particular form of interest rate that allows for the closed form representation of the optimal exercise rule.

As in the standard model the price of American put can be written as the value function of the  stopping problem
\begin{equation}\label{prob-interest}
V(t,x)=\sup_{t\le \tau\le T}\EE_{t,x}\left[e^{-\int_t^\tau r(s)ds }(K - X_\tau)^+\right]
\end{equation}
for $t\in[0,T)$ and $x>0$.
\medskip

The analysis below is straightforward extension of the constant interest case so we will focus only on main results. As we have Markovian setting, it is natural to define the exercise and waiting regions, respectively,
\begin{align}
&\mathcal{E}=\{(t,x)\in[0,T)\times (0,\infty): V(t,x)=K-x\}\\
&\mathcal{C}=\{(t,x)\in[0,T)\times (0,\infty): V(t,x)>K-x\}.
\end{align}
Now let us apply Ito-Tanaka's formula for the discounted payoff in order to gain some insight into the structure of the exercise region
\begin{align}
\EE_{t,x}\left[e^{-\int_t^\tau r(s)ds }(K - X_\tau)^+\right]
=&(K-x)^+ -\EE_{t,x}\left[\int_t^\tau e^{-\int_t^u r(s)ds }r(u)KI(X_u\le K)du\right]\\
&+\,\EE_{t,x}\left[\int_t^\tau e^{-\int_t^u r(s)ds }\delta X_uI(X_u\le K)du\right]\notag\\
&+\frac{1}{2}\EE_{t,x}\left[\int_t^\tau e^{-\int_t^u r(s)ds }d\ell_u^K(X)\right]\notag
\end{align}
where $\ell^K(X)$ is the local time that the process $X$ `spends' at $K$. The right-hand side can be explained as follows: the first term is the immediate payoff, the second and third terms are respectively the instantaneous  losses and benefits (due to postponed interests $rKdu$ and  collected dividends $\delta X_udu$, respectively) of waiting to exercise when the option is in-the-money.
Hence, one can deduce that there exists the optimal exercise boundary $b(t)$ defined on $[0,T)$ such that
\begin{align}
\mathcal{E}=\{(t,x)\in[0,T)\times (0,\infty): x\le b(t)\}.
\end{align}

2. Now we can apply the same arguments as for the case of constant interest rate to establish EEP representation for the American put price 
\begin{align}\label{ir-amer-price}
V(t,x)=&V^e(t,x)\\
&+K\int_t^T e^{-\int_t^u r(s)ds}r(u)N\left(\frac{1}{\sigma\sqrt{u-t}}\left(\log \frac{ b(u)}{x}-\int_t^u  \left(r(s)\m\delta\m \frac{\sigma^2}{2}\right)ds\right)\right)du\notag\\
&- x \int_t^T e^{-\delta(u-t)}\delta N\left(\frac{1}{\sigma\sqrt{u-t}}\left(\log \frac{ b(u)}{x}-\int_t^u  \left(r(s)\m\delta\p\frac{\sigma^2}{2}\right)ds\right)\right)du\notag
\end{align}
for all $t\in[0,T)$ and $x>0$, where $V_e$ is the European option price given as
\begin{align}
V^e(t,x)= &Ke^{-\int_t^T r(s)ds}N\left(\frac{1}{\sigma\sqrt{T-t}}\left(\log \frac{ K}{x}\m\int_t^T \left(r(s)\m\delta\m\frac{\sigma^2}{2}\right)ds\right)\right)
\\
&-xe^{-\delta(T-t)}\cdot N\left(\frac{1}{\sigma\sqrt{T-t}}\left(\log \frac{ K}{x}-\int_t^T \left(r(s)\m\delta \p\frac{\sigma^2}{2}\right)ds\right)\right)\notag
\end{align}
and $N(\cdot)$ is the cdf of $\mathcal{N}(0,1)$.
Now using the continuous pasting condition at $x=b(t)$, we obtain 
 the integral equation for the optimal exercise boundary $ b$  in \eqref{prob-interest}
\begin{align}\label{IE-interest}
K-b(t)=&V^e(t,b(t))\\
&+K\int_t^T e^{-\int_t^u r(s)ds}r(u)N\left(\frac{1}{\sigma\sqrt{u-t}}\left(\log \frac{ b(u)}{ b(t)}-\int_t^u  \left(r(s)\m\delta\m\frac{\sigma^2}{2}\right)ds\right)\right)du\notag\\
&-\delta\, b(t) \int_t^T e^{-\delta(u-t)}N\left(\frac{1}{\sigma\sqrt{u-t}}\left(\log \frac{ b(u)}{b(t)}-\int_t^u  \left(r(s)\m\delta\p\frac{\sigma^2}{2}\right)ds\right)\right)du\notag
\end{align}
for $t\in [0,T)$ with $b(T-)=K\min(1,r(T)/\delta)$.
\medskip

One can adopt the proof from \cite{Pe-2005} to show that the integral equation \eqref{IE-interest} has a unique solution in the class of continuous functions such that $\delta b(t)-r(t)K<0$ for $t\in[0,T)$. 
We will now choose the interest rate function $r(t)$ so that $r(T-)>\delta$ (i.e., $b(T-)=K$) and the integral equation \eqref{IE-interest} has explicit continuous  solution that satisfies $b(t)<\frac{r(t)}{\delta} K$ for $t\in[0,T)$.

\begin{theorem} Let us assume that the interest rate is given by
\begin{equation}\label{ir}
   r(t)=\frac{ n\left(\sqrt{(2\delta+\sigma^2)(T-t)}\right)}{1+\frac{2\sigma }{\sqrt{2\delta+\sigma^2}}\left(N\left(\sqrt{(2\delta+\sigma^2)(T-t)}\right)-\frac{1}{2}\right)}\cdot\frac{\sigma}{\sqrt{T-t}}+\delta+\frac{\sigma^2}{2}
\end{equation}
for $t\in[0,T)$ and where $n(\cdot)$ is the pdf of $\mathcal{N}(0,1)$.
Then the optimal exercise boundary has closed form 
expression
\begin{equation}\label{ir-b}
b(t)=\frac{ K}{1+\frac{2\sigma }{\sqrt{2\delta+\sigma^2}}\left(N\left(\sqrt{(2\delta+\sigma^2)(T-t)}\right)-\frac{1}{2}\right)} , \quad t\in[0,T).
\end{equation}

\end{theorem}
\begin{proof}
Let us define  
\begin{equation}
    \gamma(t)=r(t)-\delta-\sigma^2/2
\end{equation}
for $t\in[0,T)$
The main  idea is to postulate the 
 particular form of the boundary $b$ such that the integral term in \eqref{IE-interest} can simplified. This can be achieved by selecting
\begin{equation}
b(t)= K e^{-\int_t^T \gamma(s)ds}
\end{equation}
for $t\in[0,T)$  which implies
\begin{equation}
    N\left(\frac{1}{\sigma\sqrt{u-t}}\left(\log \frac{ b(u)}{ b(t)}-\int_t^u  \left(r(s)\m\delta\m\frac{\sigma^2}{2}\right)ds\right)\right)=N(0)=\frac{1}{2}
\end{equation}
for $u\in[t,T)$.

We can then rewrite the equation \eqref{IE-interest} as follows using integration by parts
\begin{align*}
 K-K e^{-\int_t^T \gamma(s)ds}=&V^{e} \left(t,K e^{-\int_t^T \gamma(s)ds}\right) +\frac{K}{2}\int_t^T e^{-\int_t^u r(s)ds}r(u)du \\
 &-\delta K e^{-\int_t^T \gamma(s)ds}\int_t^T e^{-\delta(u-t)}N\left(-\sigma\sqrt{u-t}\right)du
 \\
 =& Ke^{-\int_t^T r(s)ds} N\left(\frac{1}{\sigma\sqrt{T-t}}\left(\log \frac{ K}{ b(t)}-\int_t^T (r(s)\m\delta\m\sigma^2/2)ds\right)\right) \\ &-K e^{-\int_t^T (r(s)-\sigma^2/2)ds} N\left(\frac{1}{\sigma\sqrt{T-t}}\left(\log \frac{ K}{ b(t)}-\int_t^T (r(s)\m\delta\p\sigma^2/2)ds\right)\right) \\ &+\frac{K}{2} \left(1-e^{-\int_t^T r(s)ds}\right)\\
 &+ K e^{-\int_t^T \gamma(s)ds} e^{-\delta(T-t)}N\left(-\sigma\sqrt{T-t}\right)-\frac{K}{2} e^{-\int_t^T \gamma(s)ds}\\
 &-K e^{-\int_t^T \gamma(s)ds}\int_t^Te^{-\delta(u-t)}dN(-\sigma\sqrt{u-t})\\=& \frac{K}{2} e^{-\int_t^T r(s)ds}-K e^{-\int_t^T (r(s)-\sigma^2/2)ds} N\left(-\sigma \sqrt{T-t}\right)+\frac{K}{2} \left(1-e^{-\int_t^T r(s)ds}\right)\\ &+ K e^{-\int_t^T \gamma(s)ds} e^{-\delta(T-t)}N\left(-\sigma\sqrt{T-t}\right)-\frac{K}{2} e^{-\int_t^T \gamma(s)ds}\\
 &+K e^{-\int_t^T \gamma(s)ds}\frac{\sigma }{\sqrt{2\delta+\sigma^2}}\left(N\left(\sqrt{(2\delta+\sigma^2)(T-t)}\right)-\frac{1}{2}\right)\\=& \frac{K}{2} -\frac{K}{2} e^{-\int_t^T \gamma(s)ds}+K e^{-\int_t^T \gamma(s)ds}\frac{\sigma }{\sqrt{2\delta+\sigma^2}}\left(N\left(\sqrt{(2\delta+\sigma^2)(T-t)}\right)-\frac{1}{2}\right)
\end{align*}
so that 
\begin{equation}
  e^{\int_t^T \gamma(s)ds}=1+\frac{2\sigma }{\sqrt{2\delta+\sigma^2}}\left(N\left(\sqrt{(2\delta+\sigma^2)(T-t)}\right)-\frac{1}{2}\right)
\end{equation}
or
\begin{equation}
    \gamma(t)=\frac{ n\left(\sqrt{(2\delta+\sigma^2)(T-t)}\right)}{1+\frac{2\sigma }{\sqrt{2\delta+\sigma^2}}\left(N\left(\sqrt{(2\delta+\sigma^2)(T-t)}\right)-\frac{1}{2}\right)}\cdot\frac{\sigma}{\sqrt{T-t}}
\end{equation}
for $t\in[0,T)$. Hence, the interest rate must be given as
\begin{equation}
   r(t)=\frac{ n\left(\sqrt{(2\delta+\sigma^2)(T-t)}\right)}{1+\frac{2\sigma }{\sqrt{2\delta+\sigma^2}}\left(N\left(\sqrt{(2\delta+\sigma^2)(T-t)}\right)-\frac{1}{2}\right)}\cdot\frac{\sigma}{\sqrt{T-t}}+\delta+\frac{\sigma^2}{2}.
\end{equation}
 We note that $r(t)$ goes to $+\infty$ as $t$ approaches $T$ so that $r(T-)>\delta$ as we assumed at the beginning. Now having expression for $\gamma(t)$, the boundary $b$ is
\begin{equation}
b(t)= K e^{-\int_t^T \gamma(s)ds}=\frac{ K}{1+\frac{2\sigma }{\sqrt{2\delta+\sigma^2}}\left(N\left(\sqrt{(2\delta+\sigma^2)(T-t)}\right)-\frac{1}{2}\right)} 
\end{equation}
for  $t\in[0,T)$.
\end{proof}

\begin{remark}
We note that the interest rate function given in \eqref{ir} is increasing in $t$ on $[0,T)$ as the numerator of the first term is increasing in $t$ and the denominator is decreasing. 
\end{remark}

\begin{remark}
As can be seen from the expression \eqref{ir}, the shortcoming of the result is that the interest rate $r(t)$ goes to $+\infty$ as $t$ approaches $T$ (see also left panel of Figure \ref{int_rate}). However, if one invests in money market account at $t$, the terminal value at $T$ is still finite
\begin{equation}
 e^{\int_t^T r(s)ds}<\infty.
 \end{equation}
Despite its unrealistic behavior, this result can be still used as the closed form approximation for the market environment where the interest rate is expected to spike up during the lifetime of an option.

\end{remark}

\begin{remark}
If the dividend yield $\delta=0$, then the expressions can be simplified to
\begin{equation}
   r(t)=\frac{ \sigma n\left(\sigma\sqrt{T-t}\right)}{2\sqrt{T-t}N\left(\sigma\sqrt{T-t}\right)} +\frac{\sigma^2}{2}
\end{equation}
and
\begin{equation}\label{ir-b}
b(t)= K e^{-\int_t^T \gamma(s)ds}=\frac{ K}{2 N\left(\sigma\sqrt{T-t}\right)} 
\end{equation}
for $t\in[0,T)$.
\end{remark}

\begin{remark}
We note that when $T$ goes $+\infty$, the optimal exercise boundary converges to the  previously known optimal threshold 
$K/\left(1+\frac{\sigma }{\sqrt{2\delta+\sigma^2}}\right)$ for the perpetual put option when $r=\delta+\sigma^2/2$.
\end{remark}

\begin{remark}
Despite the fact that we have found the closed form expression for the exercise boundary $b$ in the chosen model of interest rate, the American option price is still given in the integral form \eqref{ir-amer-price}. 
\end{remark}

\begin{figure}[t]
 \begin{minipage}{0.49\linewidth}
  \centerline{\includegraphics[scale=0.9]{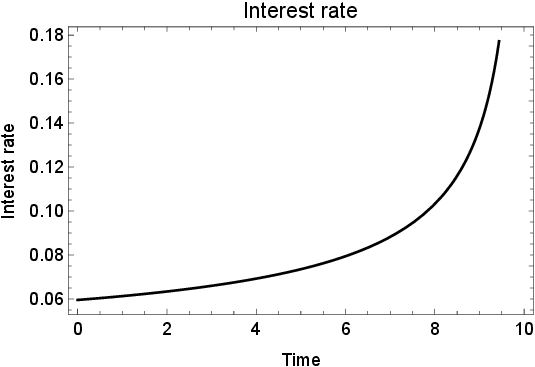}}
 \end{minipage}
\hfill
 \begin{minipage}{0.49\linewidth}
  \centerline{\includegraphics[scale=0.9]{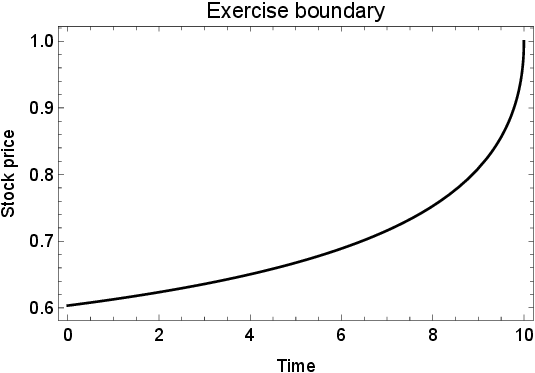}}
 \end{minipage}
 
   \caption{This figure displays the interest rate function (left)  given by \eqref{ir} and the optimal exercise boundary (right) given by \eqref{ir-b}. The parameters are $\delta=0,
 \sigma = 0.3, T = 10, K=1.$}  \label{int_rate}
   \end{figure}

\section{Time-dependent dividend yield}\label{dividend}

1. In this section we derive similar result. The aim again is to offer the model that has closed form expression for the optimal exercise boundary. The difference is that the interest rate is constant now but 
the dividend yield will be given by  particular deterministic function. 

We consider American put option problem with strike $K$ and maturity $T$
\begin{equation}\label{prob-div}
V(t,x)=\sup_{t\le \tau\le T}\EE_{t,x}\left[e^{-r(\tau-t) }(K - X_\tau)^+\right]
\end{equation}
where the asset price $X$ is given by 
\begin{equation}
dX_t/X_t=(r-\delta(t))dt+\sigma dW_t
\end{equation}
with deterministic dividend yield $\delta(t)>0$ (to be specified), constant interest rate $r>0$ and volatility
$\sigma>0$, and $W$ is a SBM under the risk-neutral measure $\QQ$.   The similar arguments from the previous section can be applied here, hence we omit details and highlight only main results and derivations. 
\medskip

2. As in the previous section, 
there exists the optimal exercise boundary $b$ and the corresponding integral equation can be written as
\begin{align}\label{IE-div}
K-b(t)=&V^e(t,b(t))\\
&+rK\int_t^T e^{-r(u-t)}N\left(\frac{1}{\sigma\sqrt{u-t}}\left(\log \frac{ b(u)}{ b(t)}-\int_t^u \left(r\m\delta(s)\m\frac{\sigma^2}{2}\right)ds\right)\right)du\notag
\\
&-b(t)\int_t^T e^{-\int_t^u \delta(s)ds}\delta (u)N\left(\frac{1}{\sigma\sqrt{u-t}}\left(\log \frac{ b(u)}{ b(t)}-\int_t^u \left(r\m\delta(s)\p\frac{\sigma^2}{2}\right)ds\right)\right)du\notag
\end{align}
for $t\in [0,T]$ with $b(T-)=K\min(1,r/\delta(T))$, where $V_e$ is the European option price 
\begin{align}
V^e(t,x)= &Ke^{-r(T-t)}N\left(\frac{1}{\sigma\sqrt{T-t}}\left(\log \frac{ K}{x}\m\int_t^T \left(r\m\delta(s)\m\frac{\sigma^2}{2}\right)ds\right)\right)
\\
&-xe^{-\int_t^T \delta(s)ds}N\left(\frac{1}{\sigma\sqrt{T-t}}\left(\log \frac{ K}{x}-\int_t^T \left(r\m\delta(s)\p\frac{\sigma^2}{2}\right)ds\right)\right)\notag
\end{align}
for $t\in[0,T)$ and $x>0$. 
\medskip

Now we will select the dividend yield function $\delta(t)$ such that $\delta(T)<r$, i.e., $b(T)=K$ and such that the boundary $b(t)$ has particular form
\begin{equation}
b(t)= K e^{-\int_t^T \gamma(s)ds}
\end{equation}
for $t\in[0,T)$, where $\gamma(t)\equiv     r-\delta(t)+\sigma^2/2$. Note that $\gamma(t)$ here has a different definition from the previous section.
\medskip

Using this form of the boundary $b$ and integration by parts, we can rewrite the equation \eqref{IE-div} as follows
\begin{align*}
 K-K e^{-\int_t^T \gamma(s)ds}=&V^{e} \left(t,K e^{-\int_t^T \gamma(s)ds}\right) +rK\int_t^T e^{-r(u-t)}N(\sigma\sqrt{u-t})du\\
 &-\frac{Ke^{-\int_t^T \gamma(s)ds}}{2}\int_t^T e^{-\int_t^u \delta(s)ds}\delta(u)du\\=& Ke^{-r(T-t)} N\left(\frac{1}{\sigma\sqrt{T-t}}\left(\log \frac{ K}{ b(t)}-\int_t^T (r\m\delta(s)\m\sigma^2/2)ds\right)\right) \\ &-b(t) e^{-\int_t^T \delta(s)ds} N\left(\frac{1}{\sigma\sqrt{T-t}}\left(\log \frac{ K}{ b(t)}-\int_t^T (r\m\delta(s)\p\sigma^2/2)ds\right)\right) \\ &+rK\int_t^T e^{-r(u-t)}N(\sigma\sqrt{u-t})du\\
 &-\frac{Ke^{-\int_t^T \gamma(s)ds}}{2} \left(1-e^{-\int_t^T \delta(s)ds}\right)\\=& Ke^{-r(T-t)}N(\sigma\sqrt{T-t})-\frac{K}{2} e^{-\int_t^T (\gamma(s)+\delta(s))ds} \\
 &-Ke^{-r(T-t)}N(\sigma\sqrt{T-t})+\frac{K}{2}
 \\
 &+K\int_t^T e^{-r(u-t)}dN(\sigma\sqrt{u-t})\\
 &-\frac{Ke^{-\int_t^T \gamma(s)ds}}{2} \left(1-e^{-\int_t^T \delta(s)ds}\right)\\
 =&\frac{K}{2}-\frac{Ke^{-\int_t^T \gamma(s)ds}}{2}+\frac{\sigma K}{\sqrt{2r+\sigma^2}}\left(N\left(\sqrt{(2r+\sigma^2)(T-t)}\right)-0.5\right)
\end{align*}
so that 
\begin{equation}
e^{-\int_t^T \gamma(s)ds}=1-\frac{2\sigma }{\sqrt{2r+\sigma^2}}\left(N\left(\sqrt{(2r+\sigma^2)(T-t)}\right)-0.5\right).
\end{equation}
Taking logarithm and then differentiating both sides, we get
\begin{equation}
    \gamma(t)=\frac{ n(\sqrt{(2r+\sigma^2)(T-t)})}{1-\frac{2\sigma }{\sqrt{2r+\sigma^2}}(N(\sqrt{(2r+\sigma^2)(T-t)})-0.5)}\cdot\frac{\sigma}{\sqrt{T-t}}
\end{equation}
and
\begin{equation}
  \delta(t)=r-\frac{ n(\sqrt{(2r+\sigma^2)(T-t)})}{1-\frac{2\sigma }{\sqrt{2r+\sigma^2}}(N(\sqrt{(2r+\sigma^2)(T-t)})-0.5)}\cdot\frac{\sigma}{\sqrt{T-t}}+\frac{\sigma^2}{2}
\end{equation}
for $t\in[0,T)$. We note that $\delta(t)$ goes to $-\infty$ as $t$ approaches $T$ so that $\delta(T-)<r$ as we assumed at the beginning. Then the exercise boundary is given by
\begin{equation}\label{b-div}
b(t)= K\left(1-\frac{2\sigma }{\sqrt{2r+\sigma^2}}\left(N\left(\sqrt{(2r+\sigma^2)(T-t)}\right)-0.5\right)\right) , \quad t\in[0,T).
\end{equation}

Using the algebraic manipulations we can show that 
$$
\delta(t)b(t)-rK<0
$$
for $t\in[0,T)$ and hence we may apply the arguments from \cite{Pe-2005} to show that $b$ given above is the unique solution to \eqref{IE-div} in the class of continuous functions that satisfy  $\delta(t)b(t)-rK<0$ for $t\in[0,T)$.

We can summarize the results of this section as follows.

\begin{theorem} Let us assume that the dividend yield is given by
\begin{equation}\label{dy}
  \delta(t)=r-\frac{ n(\sqrt{(2r+\sigma^2)(T-t)})}{1-\frac{2\sigma }{\sqrt{2r+\sigma^2}}(N(\sqrt{(2r+\sigma^2)(T-t)})-0.5)}\cdot\frac{\sigma}{\sqrt{T-t}}+\frac{\sigma^2}{2}
\end{equation}
for $t\in[0,T)$,
then the optimal exercise boundary has closed form 
expression
\begin{equation}\label{dy-b}
b(t)= K\left(1-\frac{2\sigma }{\sqrt{2r+\sigma^2}}\left(N\left(\sqrt{(2r+\sigma^2)(T-t)}\right)-0.5\right)\right) 
\end{equation}
for  $t\in[0,T)$.
\end{theorem}


\begin{remark}
The dividend yield function \eqref{dy} gets negative when $t$ is sufficiently close to $T$ and explodes to $-\infty$ at $T$. 
This can be the case for the commodities for which the storage costs will rise significantly near the end life of the option. Otherwise, we can make use of these formulas as the closed form approximations for the standard model with constant parameters. 
\end{remark}

\begin{figure}[t]
 \begin{minipage}{0.49\linewidth}
  \centerline{\includegraphics[scale=0.9]{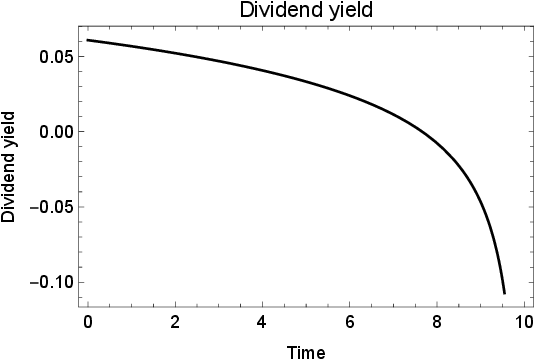}}
 \end{minipage}
\hfill
 \begin{minipage}{0.49\linewidth}
  \centerline{\includegraphics[scale=0.9]{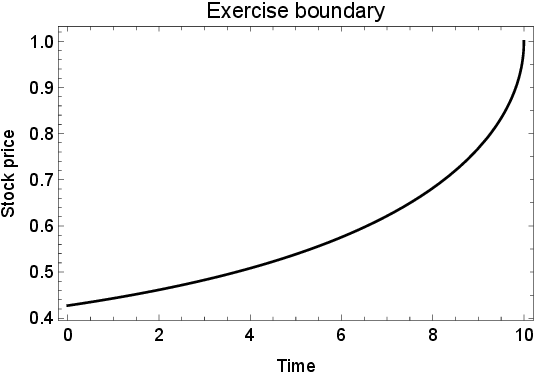}}
 \end{minipage}

   \caption{This figure displays the dividend yield  function (left)  given by \eqref{dy} and the optimal exercise boundary (right) given by \eqref{dy-b}. The parameters are $r=0.05,
 \sigma = 0.3, T = 10, K=1.$}  \label{div_yield}
   \end{figure}

\section{Time-dependent volatility}\label{volat}

In this section we assume that the interest rate is constant, the dividend yield is zero (to simplify analysis) but the volatility will be chosen as particular deterministic function of time in order to have the optimal exercise boundary in simple terms. Unlike in the previous two sections, we do not obtain the closed form for the boundary but there is simple algebraic equation that defines the boundary $b(t)$ for given $t$. Numerically it is still easier to deal with as there is no need to solve the nonlinear integral equation of Volterra type. 

Let us consider American put option problem with strike $K$ and maturity $T$
\begin{equation}\label{prob-vol}
V(t,x)=\sup_{t\le \tau\le T}\EE_{t,x}\left[e^{-r(\tau-t) }(K - X_\tau)^+\right]
\end{equation}
under the stock price model $X$
\begin{equation}
dX_t/X_t=rdt+\sigma(t) dW_t
\end{equation}
with deterministic volatility $\sigma(t)>0$ (to be specified), constant interest rate
$r>0$, and $W$ is a SBM under $\QQ$.
\medskip

We have the integral equation for the optimal exercise boundary $ b$ 
\begin{align}
K-b(t)=&V^e(t,b(t))\\
&+rK\int_t^T e^{-r(u-t)}N\left(\frac{1}{\sqrt{\int_t^u \sigma^2(u)du}}\left(\log \frac{ b(u)}{ b(t)}-\int_t^u (r\m\sigma^2(s)/2)ds\right)\right)du\notag
\end{align}
for $t\in [0,T]$ with $b(T-)=K$ and this equation has a unique solution in the class of continuous functions such that $b(t)<K$ on $[0,T)$. The European option price $V_e$  is given by
\begin{align}
V^e(t,x)= &Ke^{-r(T-t)}N\left(\frac{1}{\sqrt{\int_t^T \sigma^2(u)du}}\left(\log \frac{ K}{x}\m\int_t^T \left(r\m\frac{\sigma^2(s)}{2}\right)ds\right)\right)
\\
&-xN\left(\frac{1}{\sqrt{\int_t^T \sigma^2(u)du}}\left(\log \frac{ K}{x}-\int_t^T \left(r\p\frac{\sigma^2(s)}{2}\right)ds\right)\right)\notag
\end{align}
for $t\in[0,T)$ and $x>0$.
\bigskip

Now we will determine the volatility function $\sigma(t)$ such that the boundary has particular form
\begin{equation}
b(t)= K e^{-\int_t^T \gamma(s)ds}, \quad t\in[0,T)
\end{equation}
where $\gamma(t)=r-\sigma^2(t)/2$. 
Using this form of $b$ and integration by parts we obtain
\begin{align*}
 K-K e^{-\int_t^T \gamma(s)ds}=&V^{e} \left(t,K e^{-\int_t^T \gamma(s)ds}\right) +\frac{rK}{2}\int_t^T e^{-r(u-t)}du \\=& Ke^{-r(T-t)} N\left(\frac{1}{\sqrt{\int_t^T \sigma^2(u)du}}\left(\log \frac{ K}{ b(t)}-\int_t^T (r\m\sigma^2(s)/2)ds\right)\right) \\ &-K e^{-\int_t^T \gamma(s)ds} N\left(\frac{1}{\sqrt{\int_t^T \sigma^2(u)du}}\left(\log \frac{ K}{ b(t)}-\int_t^T (r\p\sigma^2(s)/2)ds\right)\right) \\ &+\frac{K}{2} \left(1-e^{-r(T-t)}\right)\\=& \frac{K}{2} e^{-r(T-t)}-K e^{-\int_t^T \gamma(s)ds} N\left(-\sqrt{\int_t^T \sigma^2(u)du}\right)+\frac{K}{2} \left(1-e^{-r(T-t)}\right)\\=&\frac{K}{2}-K e^{-\int_t^T \gamma(s)ds} N\left(-\sqrt{\int_t^T \sigma^2(u)du}\right)
\end{align*}
so that 
\begin{align*}
 &e^{-\int_t^T \gamma(s)ds} N\left(\sqrt{\int_t^T \sigma^2(u)du}\right)= \frac{1}{2} 
 \end{align*}
 and hence the volatility function $\sigma(t)$ satisfies the following equation
 \begin{align*}
e^{\int_t^T \frac{\sigma^2(s)}{2}ds} N\left(\sqrt{\int_t^T \sigma^2(u)du}\right)= \frac{ e^{r(T-t)}}{2} 
\end{align*}
 for $t\in[0,T)$. We can rewrite it as
  \begin{align}
e^{\frac{\phi^2(t)}{2}} N\left(\phi(t)\right)= \frac{ e^{r(T-t)}}{2}
\end{align}
where
\begin{equation}
\phi^2(t)=\int_t^T \sigma^2(u)du
\end{equation}
so that
\begin{equation}\label{vol}
  \sigma^2(t)=-2\phi(t)\phi'(t)
\end{equation}
 for $t\in[0,T)$. Thus, $\phi(t)$ can be found as the unique solution $x$ to the algebraic equation 
\begin{equation}\label{equation-phi}
    e^{\frac{x^2}{2}} N\left(x\right)= \frac{ e^{r(T-t)}}{2}
\end{equation}
for $t \in[0,T)$. It is clear that $\phi(T-)=0$ and that $\phi$ is decreasing. Once we solve the equation \eqref{equation-phi} for $\phi(t)$, we can determine the volatility $\sigma(t)$ at time $t$ using \eqref{vol}. We note that $\sigma(T-)=0$. See Figure \ref{volatility} for illustrations. 
\medskip

Hence, we derived the following result.

\begin{figure}[t]
 \begin{minipage}{0.49\linewidth}
  \centerline{\includegraphics[scale=0.9]{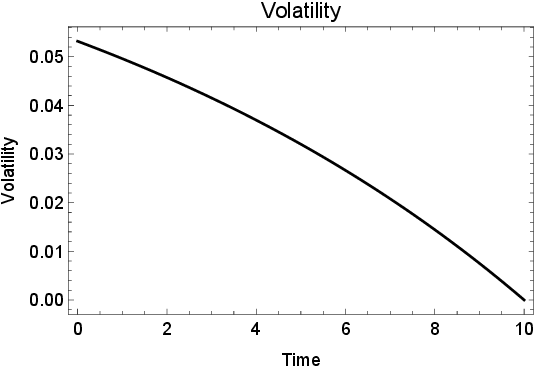}}
 \end{minipage}
\hfill
 \begin{minipage}{0.49\linewidth}
  \centerline{\includegraphics[scale=0.9]{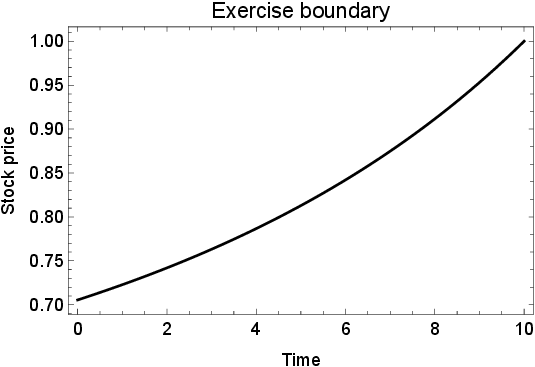}}
 \end{minipage}

   \caption{This figure displays the volatility  function (left)  given by \eqref{vol-2} and the optimal exercise boundary (right) given by \eqref{vol-b}. The parameters are $r=0.05, T = 10, K=1.$}  \label{volatility}
   \end{figure}

\begin{theorem} Let us define $\phi(t)$ as the unique solution to the algebraic equation
  \begin{align}
e^{\frac{\phi^2(t)}{2}} N\left(\phi(t)\right)= \frac{ e^{r(T-t)}}{2}
\end{align}
for $t\in[0,T)$. Then if the volatility function is defined as
\begin{equation}\label{vol-2}
  \sigma(t)=\sqrt{-2\phi(t)\phi'(t)}
\end{equation}
for $t \in[0,T)$, the optimal exercise boundary is given by 
\begin{equation}\label{vol-b}
b(t)=K e^{-r(T-t)}e^{\int_t^T \frac{\sigma^2(s)}{2}ds}=\frac{K} {2N(\phi(t))}
\end{equation}
for  $t\in[0,T)$.
\end{theorem}

\begin{remark}
From numerical experiments we observe that the volatility function given in \eqref{vol-2} is decreasing on $[0,T)$ with $\sigma(T-)=0$. This can be appropriate for situations where there is a gradual learning about the asset dynamics and hence lower uncertainty over time.  Time $t=0$ corresponds to a peak of volatility, for example, after the earnings announcement, and then volatility decreases on $[0,T)$ with extremely low values at $T$.
\end{remark}

\section{Time-dependent strike}\label{strike}

1. In this section, we have the standard  model with constant parameters. But we  consider American put option  with time-dependent strike $K(t)$ (to be specified) so that the optimal exercise boundary is given in the closed form. In other words, we stay in the Black-Scholes model but aim to choose the contract of put type with the explicit optimal exercise policy. The problem can be formulated as
\begin{equation}\label{prob-1}
V(t,x)=\sup_{t\le \tau\le T}\EE_{t,x}\left[e^{-r(\tau-t) }(K(\tau) - X_\tau)^+\right]
\end{equation}
for $t\in[0,T)$ and $x>0$,
where the asset price $X$ follows
\begin{equation}
dX_t=rX_tdt+\sigma X_tdW_t
\end{equation}
with constant interest rate
$r>0$ and volatility parameter $\sigma>0$.
\medskip

Again, the structure of the problem is similar to previous settings. There exists the optimal exercise boundary $b$. However, in this case, the local benefits of waiting to exercise are given as
\begin{equation}
    K'(t)-rK(t)
\end{equation}
for $t\in[0,T)$. If we wait at time $t$ and do not exercise, the increase $K'(t)dt$ in the strike price is instantaneous gain but there are postponed interest payments $rK(t)$ on $K$. 

Using standard arguments we derive the integral equation for the boundary $ b$ 
\begin{align}
K(t)-b(t)=&V^e(t,b(t))\\
&+\int_t^T e^{-r(u-t)}(rK(u)\m K'(u))N\left(\frac{1}{
\sigma\sqrt{u-t}}\left(\log \frac{ b(u)}{ b(t)}\m \left(r\m\frac{\sigma^2}{2}\right)(u-t)\right)\right)du\notag
\end{align}
for $t\in [0,T]$, where $V^e$ is the price of European put option with strike $K(T)$ under standard Black-Scholes model with constant parameters.
\medskip

\begin{figure}[t]
\begin{center}
 \includegraphics[scale=1]{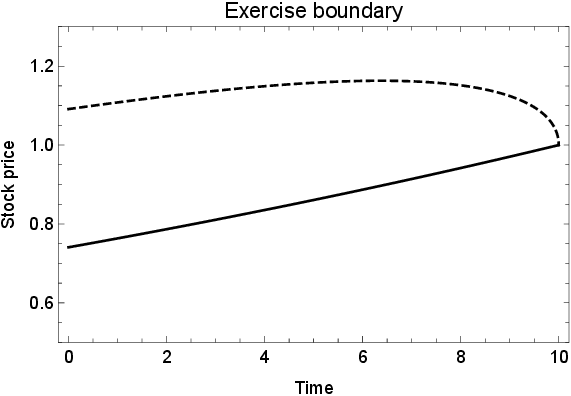}   
\end{center}

   \caption{This figure displays the strike function  $K(t)$ (dashed)  given by \eqref{strike-K} and the optimal exercise boundary $b(t)$ (solid) given by \eqref{strike-b}. The parameters are $r=0.05,
 \sigma = 0.2, T = 10, K(T)=1.$}  \label{strike-f}
   \end{figure}
   
As in the previous sections, we postulate the particular form for $b$ that makes the probability $\QQ(X_u\le b(u))$ to be 0.5. This can be achieved by
\begin{equation}\label{strike-b}
b(t)=K(T) e^{-\gamma(T-t)}
\end{equation}
for $t\in[0,T)$, where $\gamma=r-\sigma^2/2$. 
We  can then determine the strike function $K(t)$ that supports the expression for the boundary $b$ by rewriting the integral equation as follows
\begin{align*}
 K(t)-K(T) e^{-\gamma(T-t)}=&V^{e} \left(t, K(T) e^{-\gamma(T-t)}\right) \\ &+\frac{1}{2}\int_{t}^{T} e^{-r(u-t)}(rK(u)-K^{\prime}(u))  d u \\=&\frac{1}{2} K(T) e^{-r(T-t)} -K(T) e^{-\gamma(T-t)} N\left(-\sigma \sqrt{T-t}\right) \\ &-\frac{1}{2}K(T) e^{-r(T-t)}+\frac{1}{2}K(t) \\=& -K(T) e^{-\gamma(T-t)} N(-\sigma \sqrt{T-t})
 +\frac{1}{2} K(t)
\end{align*}
so that
\begin{align}\label{strike-K}
    K(t)=2 K(T) e^{-\gamma(T-t)} N(\sigma \sqrt{T-t})
\end{align}
for $t\in[0,T)$.
We can then choose any value of $K(T)$ so that we will obtain closed form expressions for $b(t)$ and $K(t)$ given by \eqref{strike-b} and \eqref{strike-K}, respectively. Figure \ref{strike-f} displays these results for given set of parameters. We note that the boundary $b$ is increasing (respectively, decreasing) if $r>\sigma^2/2$ (respectively, $r<\sigma^2/2$).
\medskip

\begin{figure}[t]
\begin{center}
 \includegraphics[scale=1]{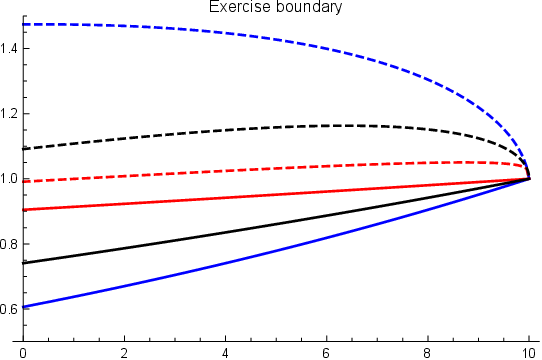}   
\end{center}

   \caption{This figure displays the strike function  $K(t)$ (dashed)  and the optimal exercise boundary $b(t)$ (solid) for different values of parameter $m$: $m=-0.02$ (red), $m=0$ (black), and $m=0.02$ (blue). The parameters are $r=0.05,
 \sigma = 0.2, T = 10, K(T)=1.$}  \label{strike-f-m}
   \end{figure}

2.  Now we will add some flexibility when define the boundary $b$. Let us assume the particular form
\begin{equation}
b(t)= K(T) e^{-\gamma(T-t)}, \quad t\in[0,T)
\end{equation}
where $\gamma=m+r-\sigma^2/2$ and $m$ is some constant that we can choose. 
After some tedious algebra, we can derive the linear integral equation of Volterra type for $K(t)$
\begin{align*}
 K(t)=\;&2 K(T) e^{-(m+r-\sigma^2/2)(T-t)} N\left(-\left(\frac{m}{\sigma}\m\sigma\right) \sqrt{T\m t}\right)\\
 &\p\frac{m}{\sigma}\int_{t}^T e^{-\left(r+\frac{m^2}{2\sigma^2}\right)(u-t)}K(u)\frac{du}{\sqrt{u-t}}
\end{align*}
for $t\in[0,T)$. Hence, given some $m$ we find corresponding $K(t)$ by solving the linear Volterra equation above, which is simpler than the standard nonlinear Volterra equation of the second kind for $b(t)$. It can be solved numerically faster than the equation for $b$. See Figure \ref{strike-f-m} for illustrations. 

Thus there is a trade-off here: if one wants to have some freedom by controlling the rate $m$, the linear integral equation for $K(t)$ must be solved numerically, and the closed form $K(t)$ exists only if $m=0$.

{}

\end{document}